\documentclass[a4paper,USenglish]{lipics-v2018}

\usepackage{amsthm}
\usepackage{hyperref}

\newcommand{\SD}{\mathcal{SD}}

\newcommand{\VD}{\mathrm{Vor}}
\newcommand{\T}{\mathcal{T}}
\newcommand{\BST}{RB-tree}
\newcommand{\PV}{K}

\newcommand{\deleted}[1]{}

\title{A Nearly Optimal Algorithm for the Geodesic Voronoi Diagram of Points in a Simple Polygon}
\titlerunning{A Nearly Optimal Algorithm for the Geodesic Voronoi Diagram in a Simple Polygon}

\author{Chih-Hung Liu}{Department of Computer Science, ETH Z\"{u}rich, Z\"{u}rich, Switzerland}{chih-hung.liu@inf.ethz.ch}{https://orcid.org/0000-0001-9683-5982}{}
\authorrunning{Chih-Hung Liu}
\Copyright{Chih-Hung Liu}

\subjclass{ Theory of computation →  Randomness, geometry and discrete structures}
\category{}

\EventEditors{Bettina Speckmann and Csaba D. T{\'o}th}
\EventNoEds{2}
\EventLongTitle{34th International Symposium on Computational Geometry (SoCG 2018)}
\EventShortTitle{SoCG 2018}
\EventAcronym{SoCG}
\EventYear{2018}
\EventDate{June 11--14, 2018}
\EventLocation{Budapest, Hungary}
\EventLogo{socg-logo} 
\SeriesVolume{99}
\ArticleNo{58} %

\nolinenumbers
\hideLIPIcs 

\begin{document}

\maketitle

 \begin{abstract}
	The \emph{geodesic Voronoi diagram} of $m$ point sites inside a simple polygon of $n$ vertices
	is a subdivision of the polygon into $m$ cells, one to each site, such that all points in a cell share the same nearest site under the geodesic distance.
	The best known lower bound for the construction time is $\Omega(n+m\log m)$, and a matching upper bound is a long-standing open question. 
	The state-of-the-art construction algorithms achieve $O\big((n+m)\log (n+m)\big)$ and $O(n+m\log m\log^2n)$ time, which are optimal for $m=\Omega(n)$ and $m=O(\frac{n}{\log^3n})$, respectively.
	In this paper, we give a construction algorithm with $O\big(n+m(\log m+\log^2 n)\big)$ time, 
	and it is \emph{nearly optimal} in the sense that if a single Voronoi vertex can be computed in $O(\log n)$ time, then the construction time will become the optimal $O(n+m\log m)$.
	In other words, we reduce the problem of constructing the diagram in the optimal time to the problem of computing a single Voronoi vertex in $O(\log n)$ time. 
	
\end{abstract}

\section{Introduction}\label{sec-ind}

The \emph{geodesic Voronoi diagram} of $m$ point sites inside a simple polygon of $n$ vertices
is a subdivision of the polygon into $m$ \emph{cells}, one to each site, 
such that all points in a cell share the same \emph{nearest site} where the distance between two points is the length of the shortest path between them inside the polygon. 
The common boundary between two cells is a \emph{Voronoi edge},
and the endpoints of a Voronoi edge are \emph{Voronoi vertices}.
A cell can be augmented into \emph{subcells} such that all points in a subcell share the same \emph{anchor},
where the anchor of a point in the cell is the vertex of the shortest path from the point to the associated site that immediately precedes the point.
An anchor is either a point site or a \emph{reflex} polygon vertex. Figure~\ref{fig-GVD-PT}(a) illustrates an augmented diagram.

The size of the (augmented) diagram is $\Theta(n+m)$~\cite{Aronov1989}.
The best known construction time is $O\big((n+m)\log(n+m)\big)$~\cite{PapadopoulouL98} and $O(n+m\log m\log^2n)$~\cite{OhA17}.
They are optimal for $m=\Omega(n)$ and for $m=O(\frac{n}{\log^3n})$, respectively, since the best known lower bound is $\Omega(n+m\log m)$.
The existence of a matching upper bound is a long-standing open question by Mitchell~\cite{Mitchell00}.

Aronov~\cite{Aronov1989} first proved  fundamental properties: 
a bisector between two sites is a simple curve consisting  of $\Theta(n)$ straight and hyperbolic arcs and ending on the polygon boundary; 
the diagram has $\Theta(n+m)$ vertices, $\Theta(m)$ of which are Voronoi vertices.
Then, he developed a divide-and-conquer algorithm that recursively partitions the polygon into two roughly equal-size sub-polygons.
Since each recursion level takes $O\big((n+m)\log(n+m)\big)$ time to extend the diagrams between every pair of sub-polygons,
the total time is $O\big((n+m)\log(n+m)\log n\big)$.

Papadopoulou and Lee~\cite{PapadopoulouL98} combined the divide-and-conquer and plane-sweep paradigms to improve the construction time to $O\big((n+m)\log(n+m)\big)$. First, the polygon is triangulated and one resultant triangle is selected as the root such that the dual graph is a rooted binary tree and each pair of adjacent triangles have a parent-child relation; see Figure~\ref{fig-GVD-PT}(b).
For each triangle, its diagonal shared with its parent partitions the polygon into two sub-polygons: the ``\emph{lower}'' one contains it, and the ``\emph{upper}'' one contains its parent. 
Then, the triangles are swept by the post-order and pre-order traversals of the rooted tree to respectively build, inside each triangle, the two diagrams with respect to sites in its lower and upper sub-polygons. Finally, the two diagrams inside each triangle are merged into the final diagram.

Very recently, Oh and Ahn~\cite{OhA17} generalized the notion of plane sweep to a simple polygon.
To supplant the scan line,  one point is fixed on the polygon boundary, 
and another point is moved from the fixed point along the polygon boundary counterclockwise,
so that the shortest path between the two points will sweep the whole polygon.
Moreover, Guibas and Hershberger's data structure for shortest path queries~\cite{Guibas1989, Hershberger91} is extended to compute a Voronoi vertex among three sites or between two sites in $O(\log^2 n)$ time.
This technique enables handling an event in $O(\log m \log^2 n)$ time,
leading to a total time of $O(n+m\log m\log^2n)$.

Papadopoulou and Lee's method~\cite{PapadopoulouL98} has two issues inducing the $n\log(n+m)$ time-factor.
First, while sweeping the polygon,
an intermediate diagram is represented by a ``\emph{wavefront}''  in which a ``\emph{wavelet}'' is associated with a ``subcell.'' 
Although this representation enables computing the common vertex among three subcells in $O(1)$ time,
since it takes~$\Omega\big(\log (n+m)\big)$ time to update such a wavefront,
the $\Omega(n)$ vertices lead to the $n\log (n+m)$ factor.
Second, when a wavefront enters a triangle from one diagonal
and leaves from the other two diagonals, it will split into two. 
Since there are $\Omega(n)$ triangles, there are $\Omega(n)$ split events,
and since a split event takes $\Omega\big(\log (n+m)\big)$ time, the $n\log (n+m)$ factor arises again.

\subsection{Our contribution}\label{sub-ind-contribution}


We devise a construction algorithm with $O\big(n+m(\log m+\log^2 n)\big)$ time, which is slightly faster than Oh and Ahn's method~\cite{OhA17}
and is optimal for $m=O(\frac{n}{\log^2 n})$.
More importantly,  
our algorithm is, to some extent, \emph{nearly optimal} since the $\log^2 n$ factor solely comes from computing a single Voronoi vertex.
If the computation time can be improved to $O(\log n)$, the total construction time will become $O\big(n+m(\log m+\log n)\big)$, which equals the optimal $O(n+m\log m)$ since $m\log n=O(n)$ for $m=O(\frac{n}{\log n})$ and $\log n=O(\log m)$ for $m=\Omega(\frac{n}{\log n})$. 
In other words, we reduce the problem of constructing the diagram in the optimal time by Mitchell~\cite{Mitchell00}
to the problem of computing a single Voronoi vertex in $O(\log n)$ time.

At a high level, our algorithm is a new implementation of Papadopoulou and Lee's concept~\cite{PapadopoulouL98}
using a different data structure of a wavefront, symbolic maintenance of incomplete Voronoi edges, tailor-made wavefront operations, and appropriate amortized time analysis. 

First, in our wavefront, each wavelet is directly associated with a cell rather than a subcell.
This representation makes use of Oh and Ahn's~\cite{OhA17} $O(\log^2 n)$-time technique of computing a Voronoi vertex.
Each wavelet also stores the anchors of incomplete subcells in its associated cell in order to enable locating a point in a subcell along the wavefront.

Second, if each change of a wavefront needs to be updated immediately, 
a priority queue for events would be necessary, and since the diagram has $\Theta(m+n)$ vertices,  an $(m+n)\log (m+n)$ time-factor would be inevitable. To overcome this issue, 
we maintain incomplete Voronoi edges symbolically, and update them only when necessary. For example, 
during a binary search along a wavefront, each incomplete Voronoi edge of a tested wavelet will be updated.

Third, to avoid $\Omega(n)$ split operations, 
we design two tailor-made operations. If a wavefront will separate into two but one part will not be involved in the follow-up ``sweep'', then, 
instead of using a binary search, we ``divide'' the wavefront in a seemingly brute-force way in which we traverse the wavefront from the uninvolved part until the ``division'' point,
remove all visited subcells, and build another wavefront from those subcells.
If a wavefront propagates into a sub-polygon that contains no point site,
then we adopt a two-phase process to build the diagram inside the sub-polygon instead of splitting a wavefront many times.

Finally, when deleting or inserting a subcell (anchor), its position in a wavelet is  known.
Since re-balancing a \emph{red-black tree} (RB-tree) after an insertion or a deletion takes amortized $O(1)$ time~\cite{MehlhornS08,Tarjan83,Tarjan85}, 
by augmenting each tree node with pointers to its predecessor and successor, an insertion or a deletion with a known position takes amortized~$O(1)$ time.  


This paper is organized as follows. 
Section~\ref{sec-pre} formulates the geodesic Voronoi diagram,
defines a rooted partition tree, and introduces Papadopoulou and Lee's two subdivisions~\cite{PapadopoulouL98};
Section~\ref{sec-overview} summarizes our algorithm;
Section~\ref{sec-wavefront} designs the data structure of a wavefront; 
Section~\ref{sec-operation} presents wavefront operations; Section~\ref{sec-construction} implements the algorithm with those operations.

 \section{Preliminary}\label{sec-pre}

\subsection{Geodesic Voronoi diagrams}\label{sub-pre-GVD}

Let $P$ be a simple polygon of $n$ vertices, let $\partial P$ denote the boundary of $P$, and let $S$ be a set of $m$ point sites inside $P$. For any two points $p,q$ in $P$,
the \emph{geodesic distance} between them, denoted by $d(p, q)$,
is the length of the shortest path between them  that fully lies in $P$,
the \emph{anchor} of $q$ with respect to $p$
is the last vertex on the shortest path from $p$ to $q$ before $q$,
and the \emph{shortest path map} (SPM) from $p$ in $P$ is a subdivision of $P$ such that all points in a region share the same anchor with respect to $p$. 
Each edge in the SPM from $p$ is a line segment from a reflex polygon vertex $v$ of $P$ to $\partial P$
along the direction from the anchor of $v$ (with respect to $p$) to $v$, and this line segment is called the \emph{SPM edge} of $v$ (from $p$). 

The \emph{geodesic Voronoi diagram} of $S$ in $P$, denoted by $\VD_P(S)$, partitions $P$ into $m$ \emph{cells},
one to each site,
such that all points in a cell share the same nearest site in $S$ under the geodesic distance. 
The cell of a site $s$ can be augmented by partitioning
the cell with the SPM from $s$ into \emph{subcells} such that all points in a subcell share the same anchor with respect to $s$. The augmented version of $\VD_P(S)$ is denoted by $\VD_P^*(S)$.
With a slight abuse of terminology, a cell in $\VD_P^*(S)$ indicates a cell in $\VD_P(S)$ together with its subcells in $\VD_P^*(S)$. Then, each cell is associated with a site,
and each subcell is associated with an achor.
As shown in Figure~\ref{fig-GVD-PT}(a), $v$ is the anchor of the shaded subcell (in $s_1$'s cell), and the last vertex on the shortest path from $s_1$ to any point $x$ in the shaded subcell before $x$ is $v$.

\begin{figure}
	\includegraphics{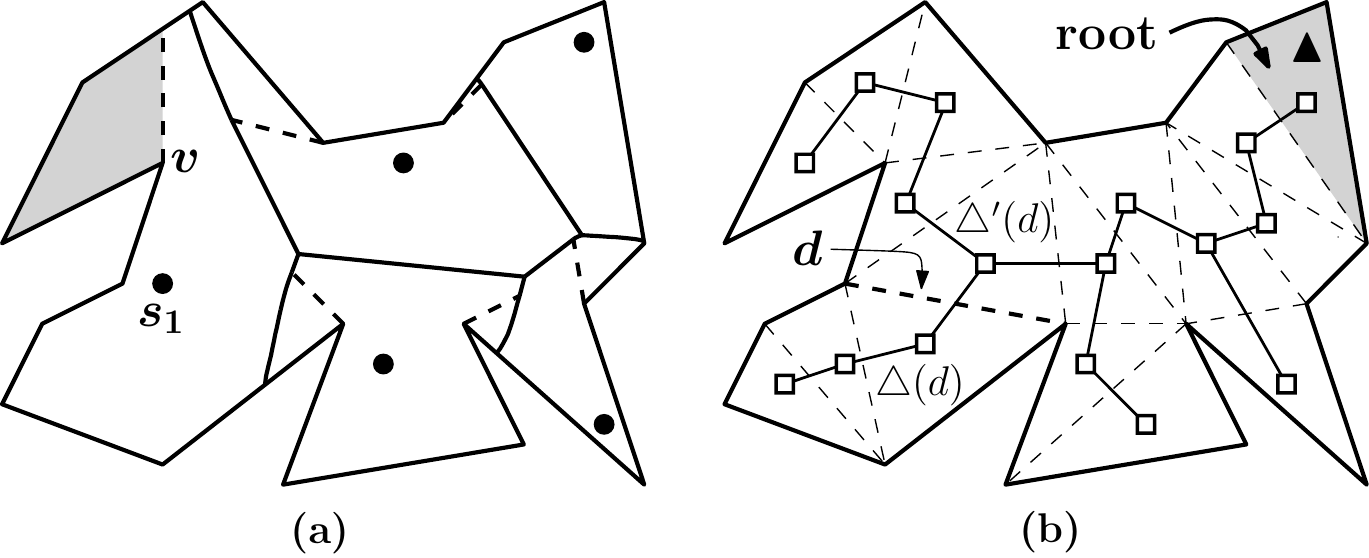}
	\caption{(a) Augmented geodesic Voronoi diagram $\VD_P^*(S)$. (b) Rooted partition tree $\T$.}\label{fig-GVD-PT}
\end{figure}

A \emph{Voronoi edge} is the common boundary between two adjacent cells,
and the endpoints of a Voronoi edge are called \emph{Voronoi vertices}. 
A Voronoi edge is a part of the \emph{geodesic bisector} between the two associated sites,
and consists of straight and hyperbolic arcs. Endpoints of these arcs except Voronoi vertices are called \emph{breakpoints}, and a breakpoint is incident to an SPM edge in the SPM from one of the two associated sites, indicating a change of the corresponding anchor.
There are $\Theta(m)$ Voronoi vertices and $\Theta(n)$ breakpoints~\cite{Aronov1989}.

In our algorithm, each anchor $u$ refers to either a reflex polygon vertex of $P$ or a point site in $S$;
we store its associated site $s$, its geodesic distance from $s$,
and its anchor with respect to $s$. The weighted distance from $u$ to a point $x$ 
is $d(s, u)+|\overline{ux}|$.

Throughout the paper, we make a general position assumption that no polygon vertex is equidistant from two sites in $S$ and no point inside $P$ is equidistant from four sites in $S$.
The former avoids nontrivial overlapping among cells \cite{Aronov1989}, and the latter ensures that the degree of each Voronoi vertex with respect to Voronoi edges is either 1 (on $\partial P$) or 3 (among three cells).

The boundary of a cell except Voronoi edges are polygonal chains on $\partial P$. For convenience,
these polygonal chains are referred to as \emph{polygonal edges} of the cell,
the incidence of an SPM edge onto a polygonal edge is also a \emph{breakpoint},
and the polygonal edges including their polygon vertices and breakpoints also count for the size of the cell.

\begin{lemma}(\cite[Lemma~5 and 14]{OhA17})\label{lem-Voronoi-vertex}
	It takes $O(\log^2 n)$ time to compute the degree-1 or degree-3 Voronoi vertex between two sites or among three sites after $O(n)$-time preprocessing.
\end{lemma}

\subsection{A rooted partition tree}\label{sub-pre-triangulation}

Following Papadopoulou and Lee~\cite{PapadopoulouL98},
a rooted partition tree $\T$ for $P$ and  $S$ is built as follows:
First, $P$ is triangulated using Chazelle's algorithm \cite{Chazelle91} in $O(n)$ time,
and all sites in $S$ are located in the resulting triangles by Edelsbrunner et al's approach~\cite{EdelsbrunnerGS86} in $O(n+m\log n)$ time. 
The dual graph of the triangulation is a tree in which each node corresponds to a triangle
and an edge connects two nodes if and only if their corresponding triangles share a diagonal. 
Then, an arbitrary triangle $\blacktriangle$ whose two diagonals are polygon sides, i.e., a node with degree 1,
is selected as the \emph{root}, so that there is a parent-child relation between each pair of adjacent triangles. 
Figure~\ref{fig-GVD-PT}(b) illustrates a rooted partition tree $\T$.

For a diagonal $d$, let $\triangle(d)$ and $\triangle'(d)$ be the two triangles adjacent to $d$ such that $\triangle'(d)$ is the parent of $\triangle(d)$,
and call $d$ the \emph{root diagonal} of $\triangle(d)$; also see Figure~\ref{fig-GVD-PT}(b).
$d$ partitions $P$ into two sub-polygons:
$P(d)$ contains $\triangle(d)$
and $P'(d)$ contains $\triangle'(d)$.
$P(d)$ and $P'(d)$ are said to be ``below'' and ``above'' $d$, respectively.
Assume that $d\subseteq P(d)$; let $S(d)=S\cap P(d)$ and $S'(d)=S\setminus S(d)$, which indicate the two respective subsets of sites below and above $d$.

In this paper,  we adopt the following convention: for each triangle $\triangle$,
let $d=\overline{v_1v_2}$ be its root diagonal, let $\triangle'$ be its parent triangle,
let $d_1,d_2$ be the other two diagonals of $\triangle$, and let $d_3, d_4$ be the other two diagonals of $\triangle'$.
Assume that $d_4$ is the diagonal between $\triangle'$ and its parent, and that $d_1$, $d$, and $d_4$ are incident to $v_1$,
and  $d_2$, $d$, and $d_3$ are incident to $v_2$.
Let $v_{1,2}$ be the vertex shared by $d_1$ and $d_2$, and let $v_{3,4}$ be the vertex shared by $d_3$ and $d_4$.
Denote the set of sites in $\triangle$ as $S_\triangle=S(d)-S(d_1)-S(d_2)$. Figure~\ref{fig-Triangle-SD}(a) shows an illustration.

 \begin{figure}
	\includegraphics{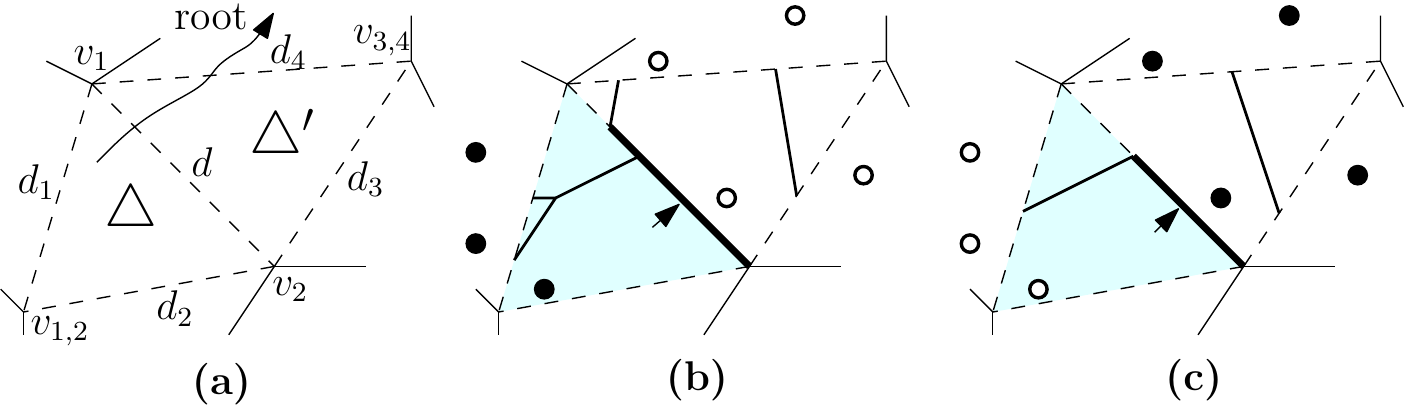}
	\caption{(a)  $\triangle$ and $\triangle'$. (b) $\SD\cap \triangle$. (c) $\SD' \cap \triangle$. (Borders on $d$ are indicated by arrows.)}\label{fig-Triangle-SD}
\end{figure}

\subsection{Subdivisions}\label{sub-pre-SD}

Papadopoulou and Lee~\cite{PapadopoulouL98} introduced two subdivisions, $\SD$ and $\SD'$, of $P$, which can be merged into $\VD_P^*(S)$.  
For each triangle $\triangle$ with a root diagonal $d$,
$\SD$ and $\SD'$ respectively contain $\VD_P^*\big(S(d)\big)\cap \triangle$ and $\VD_P^*\big(S'(d)\big)\cap \triangle$; $\SD$ also contains $\VD_P^*(S)\cap \blacktriangle$.
Since $S(d)$ and $S(d_4)$ (resp.~$S'(d)$ and $S'(d_4)$) may differ,
a \emph{border} forms along $d$ in $\SD$ (resp.~$\SD'$) to ``remedy'' the conflicting proximity information.
Figure~\ref{fig-Triangle-SD}(b)--(c) illustrate $\SD$ and $\SD'$.

The incidence of a Voronoi edge or an SPM edge in $\VD_P^*(S)$ onto a border in $\SD$ or $\SD'$ is called a \emph{border vertex},
and the border vertices partition a border into \emph{border edges}. Both $\SD$ and $\SD'$ have $O(n+m)$ border vertices \cite{PapadopoulouL98}, $O(m)$ of which are induced by Voronoi edges. 
Hereafter, a \emph{diagram vertex} means a Voronoi vertex, a breakpoint, a border vertex, or a polygon vertex.

\section{Overview of the algorithm}\label{sec-overview}

We compute $\VD_P^*(S)$ in the following three steps:
\begin{enumerate}
	\item Build the rooted partition tree $\T$ for $P$ and $S$ in $O(n+m\log n)$ time. (Section~\ref{sub-pre-triangulation})
	\item Construct $\SD$ and $\SD'$ in $O\big(n+m(\log m+\log^2 n)\big)$ time
	by sweeping the polygon using the post-order and pre-order traversals of $\T$,
	respectively. (Section~\ref{sub-postorder} and Section~\ref{sub-preorder})
	\item Merge $\SD$ and $\SD'$ into $\VD^*_P(S)$ in $O(n+m)$ time using Papadopoulou and Lee's method~\cite[Section 7]{PapadopoulouL98}.
\end{enumerate}

By the above-mentioned running times, we conclude the total running time as follows. 
\begin{theorem}\label{thm-total-time}
	$\VD_P^*(S)$ can be constructed in $O\big(n+m(\log m+\log^2 n )\big)$ time.
\end{theorem}

\section{Wavefront structure}\label{sec-wavefront}

A \emph{wavefront} represents the ``\emph{incomplete}'' boundary of ``\emph{incomplete}'' Voronoi cells during the execution of  our algorithm,
and wavefronts will ``\emph{sweep}'' the simple polygon $P$ triangle by triangle to construct $\SD$ and $\SD'$.
To avoid excessive updates, each ``\emph{incomplete}'' Voronoi edge, which is a part of a Voronoi edge and will be completed during the sweep, 
is maintained \emph{symbolically},
preventing an extra $\log n$ time-factor. 
During the sweep, candidates for Voronoi vertices in $\SD$ and $\SD'$ called \emph{potential vertices} 
will be generated in the \emph{unswept} part of $P$.

\subsection{Formal definition and data structure}\label{sub-wavefront-DS}

Let $\eta$ be a diagonal or a pair of diagonals sharing a common polygon vertex,
and let $S'$ be a subset of $S$ lying on the same side of $\eta$.
A \emph{wavefront} $W_{\eta}(S')$ represents the sequence of Voronoi cells in $\VD_P^*(S')$ appearing along $\eta$,
and each appearance of a cell induces a \emph{wavelet} in $W_{\eta}(S')$.
The \emph{unswept} area of $W_{\eta}(S')$ is the part of $P$ on the opposite side of $\eta$ from $S'$.
Since $\VD_P^*(S')$ in the unswept area has not yet been constructed, Voronoi and polygonal edges incident to $\eta$ are called \emph{incomplete}. 
Each wavelet is bounded by two \emph{incomplete} Voronoi or polygonal edges along $\eta$,
and its \emph{incomplete boundary} comprises its two incomplete edges and the portion of $\eta$ between them. 
When the context is clear, a wavelet may indicate its associated cell.

$W_{\eta}(S')$ is stored in an \BST{} in which
each node refers to one wavelet 
and the ordering of nodes follows their appearances along $\eta$.
The \BST{} is augmented such that each node has pointers to its predecessor and successor,
and the root has pointers to the first and last nodes, enabling efficiently traversing wavelets along $\eta$
and accessing the two ending wavelets. 

The subcells of a wavelet are the subcells in its associated cell incident to its incomplete boundary. The list of their anchors is also maintained by an augmented \BST{} in which their ordering follows their appearances along the incomplete boundary. Due to the visibility of a subcell, each subcell appears exactly once along the incomplete boundary.
Since the rebalancing after an insertion or a deletion takes amortized $O(1)$ time~\cite{Tarjan83,Tarjan85,MehlhornS08},
inserting or deleting an anchor at a known position in an \BST{}, i.e., without searching, takes amortized $O(1)$ time. 

\subsection{Incomplete Voronoi and polygonal edges}\label{sub-wavefront-symbolic}

When a wavefront moves into its \emph{unswept} area,
incomplete Voronoi edges will extend, generating new breakpoints.
If each breakpoint needs to be created immediately,
all candidates for breakpoints should be maintained in a priority queue, leading to an $\Omega(n\log n)$ running time due to $\Omega(n)$ breakpoints. 
To avoid these excessive updates,
we maintain each incomplete Voronoi edge \emph{symbolically},
and update it only when necessary.  
For example, when searching along the wavefront or merging two wavefronts,
the incomplete Voronoi edges of each involved wavelet (cell) will be updated until the diagonal or the pair of diagonals.

Since a breakpoint indicates the change of a corresponding anchor, 
for a Voronoi edge,
if the anchors of its incident subcells on ``each side'' are stored in a sequence,
the Voronoi edge can be computed in time proportional to the number of breakpoints by scanning the two sequences \cite[Section 4]{OhA17}.
Following this concept, 
for each incomplete Voronoi edge,
we store its fixed Voronoi vertex (in the swept area), its last created breakpoint, and its last used anchors on its two sides, 
so that we can update an incomplete Voronoi edge by scanning the two lists of anchors from the last used anchors.
When creating a breakpoint, 
we also build a corresponding SPM edge, and then remove the corresponding anchor from the anchor list.

Each polygonal edge is also maintained symbolically in a similar way; in particular, it will also be  updated when a polygon vertex is inserted as an anchor. 
Meanwhile,  
the SPM edges incident to a polygonal edge will also be created using its corresponding anchor list.

\subsection{Potential vertices}\label{sub-wavefront-potential}

We process incomplete Voronoi edges to generate candidates for Voronoi vertices called \emph{potential vertices}. 
For each incomplete Voronoi edge, since its two associated sites lie in the swept area,
one endpoint of the corresponding bisector lies in the unswept area and is a degree-1 potential vertex.
For each two adjacent Voronoi edges along the wavefront,
their respective bisectors may intersect in the unswept area,
and the intersection is a degree-3 potential vertex.
By Lemma~\ref{lem-Voronoi-vertex}, a potential vertex can be computed in $O(\log^2 n)$ time.

Potential vertices are stored in their located triangles;
each diagonal of a triangle is equipped with a priority queue to store the potential vertices in the triangle associated with sites on the opposite side of the diagonal, where the key is the distance to the diagonal.

\section{Wavefront operations}\label{sec-operation}

We summarize the eight wavefront operations, where $\PV_{\mbox{inv}}$, $A_{\mbox{vis}}$ and $I_{\mbox{new}}$ are the numbers of involved (i.e., processed and newly created) potential vertices, visited anchors, and created diagram vertices, respectively:
\begin{description}
	\item[Initiate:] Compute $\VD_\triangle(S_\triangle)$ and initialize $W_{(d,d_2)}(S_\triangle)$ in $O\big(|S_\triangle|(\log m+\log^2 n)\big)$ time.
	\item[Extend:] Extend one wavefront into a triangle from one diagonal to the opposite pair of diagonals to build the diagram inside the triangle in  $O\big(\PV_{\mbox{inv}}(\log m+\log^2 n)+I_{\mbox{new}}\big)$ plus amortized $O(1)$ time.
	\item[Merge:] Merge two wavefronts sharing the same diagonal or the same pair of diagonals together with merging the two corresponding diagrams in $O\big(|S_\triangle|(\log m+\log n)+\PV_{\mbox{inv}}(\log m+\log^2 n)+I_{\mbox{new}}\big)$ plus amortized $O(1)$ time
	where $\triangle$ is the underlying triangle. 
	\item[Join:] Join two wavefronts sharing the same diagonal with building the border on the diagonal but without merging the two corresponding diagrams inside the underlying triangle $\triangle'$ in $O\big(|S_{\triangle'}|(\log m+\log n)+\PV_{\mbox{inv}}(\log m+\log^2 n)+I_{\mbox{new}}\big)$ plus amortized $O(1)$ time.
	\item[Split:] Split a wavefront using a binary search in $O(\log m+\log n)$ time. 
	\item[Divide:] Divide a wavefront by traversing one diagonal in amortized $O(A_{\mbox{vis}}+1)$ time.
	\item[Insert:] Insert $S_\triangle$ into $W_{d_1}\big(S(d_2)\big)$ to be $W_{d_1}\big(S(d_2)\cup S_\triangle \big)$ in $O\big(|S_\triangle|(\log m+\log^2 n)\big)$ time.
	
	\item[Propagate:] Propagate $W_{d}\big(S'(d)\big)$ into $P(d)$, provided that $P(d)\cap S=S(d)=\emptyset$,  to build $\SD'\cap P(d)=\VD_P^*(S)\cap P(d)$ in $O\big(\PV_{\mbox{inv}}(\log m+\log^2n)+|\SD'\cap P(d)|\big)$ time.
\end{description}
Merge and Join operations differ in that the former also merges the two corresponding Voronoi diagrams in the underlying triangle, and the latter does not.

Readers could directly read the algorithm in Section~\ref{sec-construction} without knowing the detailed implementation for wavefront operations.
For the operation times, we have three remarks.
\begin{remark}
	During Extend and Propagate operations and at the end of the other operations, 
	potential vertices will be generated according to new incomplete Voronoi edges.
	It will be clear in Section~\ref{sec-construction} that the total number of potential vertices in the construction of $\SD$ and $\SD'$ is $O(m)$,
	so a priority queue takes $O(\log m)$ time for an insertion or an extraction. 
\end{remark}

\begin{remark}
	As stated in Section~\ref{sub-wavefront-symbolic},
	we maintain incomplete Voronoi/polygonal edges symbolically. 
	For the sake of simplicity, we charge the time to update an incomplete edge to the created breakpoints, and assign the corresponding costs to their located triangles. 
\end{remark}

\begin{remark}
	Since a wavelet (resp.~anchor) to remove from a wavefront (resp.~wavelet) must be inserted beforehand,
	we charge its removal cost at its insertion time.
	For a wavelet, the cost is $O(\log m)$, and for an anchor, since the position is always known, 
	the cost is amortized $O(1)$. 
	Similarly, we charge the cost to delete a diagram vertex at its creation time.

\end{remark}

\subsection{Initiate operation}\label{sub-op-initiate}

An Initiate operation processes a triangle $\triangle$ to compute $\VD_\triangle(S_\triangle)$ and initiate $W_{(d,d_2)}(S_\triangle)$.
Since no polygon vertex lies in a triangle, 
each resulting Voronoi cell has only one subcell. 
For the sake of simplicity, we assume that no diagonal of $\triangle$ is a polygon side;
the other cases are similar.

An Initiate operation consists of two simple steps:
First, $\VD_\triangle(S_\triangle)$ is computed by constructing the Voronoi diagram of $S_\triangle$ without considering $P$
and trimming the diagram by $\triangle$. 
Second, 
by tracing $\VD_\triangle(S_\triangle)$ along the pair $(d,d_2)$ of diagonals,
$W_{(d,d_2)}(S_\triangle)$ is built as the aforementioned data structure (Section~\ref{sub-wavefront-DS}).  
During the tracing, if $v_{1,2}$ (resp.~$v_1$) is a reflex polygon vertex, then it will be inserted into the anchor list of its located wavelet (cells) in $W_{(d,d_2)}(S_\triangle)$. Note that if $v_2$ is a reflex polygon vertex, then $v_2$ will be inserted as an anchor after $W_{(d,d_2)}(S_\triangle)$ has been divided or split at $v_2$.

The operation time is $O\big(|S_\triangle|\cdot (\log m+\log^2n)\big)$. 
The first step takes $O(|S_\triangle|\log |S_\triangle|)$ time to construct the Voronoi diagram using
a standard divide-and-conquer or a plane-sweep algorithm \cite[Section~3]{AurenhammerKL13}, and takes $O(|S_\triangle|)$ time to trim the diagram.
The tracing in the second step simply takes $O(|S_\triangle|)$ time by first locating $v_{1,2}$ in $V_\triangle(S_\triangle)$,
and then repeatedly traversing the boundary of the located Voronoi cell to find the next cell along $(d,d_2)$ until reaching $v_1$.
The data structure can also be built from scratch in $O(|S_\triangle|)$ time. 
Since a cell in $\VD_\triangle(S_\triangle)$ has only one subcell, whose anchor is exactly the associated site,
the position to insert $v_1$ or $v_{1,2}$ into a wavelet as an anchor
can be computed in $O(1)$ time,
and although inserting an anchor at a known position takes amortized $O(1)$ time, since the worst-case time to insert an anchor is $O(\log n)$,
the time to insert $v_1$ and $v_{1,2}$ is dominated by the time-factor $O(|S_\triangle|\cdot (\log m+\log^2n)\big)$.
Finally, since there are $O(|S_\triangle|)$ incomplete Voronoi edges, there are $O(|S_\triangle|)$ potential vertices, leading to the time factor $|S_\triangle|\cdot (\log m+\log^2n)$. (Computing a potential vertex takes $O(\log^2 n)$ time (Lemma~\ref{lem-Voronoi-vertex}), locating it in a triangle of $\T$ takes $O(\log n)$ time, and inserting it in the priority queue takes $O(\log m)$ time.)

\subsection{Extend operation}\label{sub-op-extend}

An Extend operation extends a wavefront  $W_{\tilde{d}}(Q)$ from one diagonal $\tilde{d}$ of a triangle $\tilde{\triangle}$ to the opposite pair of diagonals $(\tilde{d}_1, \tilde{d}_2)$  to construct $W_{(\tilde{d}_1, \tilde{d}_2)}(Q)$ and $\VD_P^*(Q)\cap \tilde{\triangle}$,
where $\tilde{d}=\overline{\tilde{v}_1\tilde{v}_2}$, $\tilde{d}_1=\overline{\tilde{v}_1\tilde{v}_{1, 2}}$, and $\tilde{d}_2=\overline{\tilde{v}_2\tilde{v}_{1, 2}}$, and $Q$ lies on the opposite side of $\tilde{d}$ from $\tilde{\triangle}$.

This operation is equivalent to sweeping the triangle with a scan line parallel to $\tilde{d}$
and processing each hit potential vertex. 
The next hit potential vertex is provided from the priority queue associated with $\tilde{d}$ (defined in Section~\ref{sub-wavefront-potential}),
and will be processed in three phases:
First, its validity is verified: for a degree-1 potential vertex, its incomplete Voronoi edge should be alive in the wavefront,
and for a degree-3 one, its two incomplete Voronoi edges should be still adjacent in the wavefront. 
Second, the one or two incomplete Voronoi edges are updated up to the potential vertex. Since a degree-1 potential vertex is incident to a polygonal edge, 
the polygonal edge is also updated.

Finally, 
when a potential vertex becomes a Voronoi vertex,
a wavelet will be removed from the wavefront. 
For a degree-1 potential vertex, since the wavelet lies at one end of the wavefront,
a polygonal edge is added to its adjacent wavelet; for a degree-3 potential vertex, since the removal makes two wavelets adjacent, 
a new incomplete Voronoi edge is created, and one degree-1 and two degree-3 potential vertices are computed accordingly. 

After the extension, if $\tilde{v}_{1, 2}$ is a reflex polygon vertex,
the wavefront is further processed by three cases.
If neither $\tilde{d}_1$ nor $\tilde{d}_2$ is a polygon side,
$\tilde{v}_{1,2}$ will later be inserted as an anchor while dividing or splitting $W_{(\tilde{d}_1, \tilde{d}_2)}(Q)$ at $\tilde{v}_{1,2}$.
If both $\tilde{d}_1$ and  $\tilde{d}_2$ are polygon sides, 
all the subcells will be completed along $(\tilde{d}_1, \tilde{d}_2)$ and the wavefront will disappear. 
If only $\tilde{d}_1$ (resp.~$\tilde{d}_2$)  is a polygon side,
all the subcells along $\tilde{d}_1$ (resp.~$\tilde{d}_2$) excluding the one containing $\tilde{v}_{1,2}$ will be completed,
and $\tilde{v}_{1,2}$ will be inserted into the corresponding wavelet as the last or first anchor.

The operation time is $O\big(\PV_{\mbox{inv}}(\log m+\log^2 n)+I_{\mbox{new}} \big)$ plus amortized $O(1)$,
where $\PV_{\mbox{inv}}$ is the number of involved (i.e., processed and newly created) potential vertices and $I_{\mbox{new}}$ is the number of created diagram vertices.
First, extracting a potential vertex from the priority queue takes $O(\log m)$ time, and verifying its validity takes $O(1)$ time.
Second, updating incomplete Voronoi and polygonal edges takes time linear in the number of created breakpoints (Section~\ref{sub-wavefront-symbolic}), i.e., $O(I_{\mbox{new}})$ time in total.
Third, since computing a potential vertex takes  $O(\log^2 n)$ time, locating it takes $O(\log n)$ time, and inserting it into a priority queue takes $O(\log m)$ time,
creating a new potential vertex takes $O(\log^2n+\log m)$ time. 
Finally, completing subcells takes $O(I_{\mbox{new}})$ time, and inserting $\tilde{v}_{1,2}$ takes amortized $O(1)$ time since its position in the anchor list is known.

\subsection{Merge operation}\label{sub-op-merge}

A Merge operation merges $W_\eta(Q)$ and $W_\eta(Q')$ into $W_\eta(Q\cup Q')$ together with 
$\VD_P^*(Q)\cap \triangle$ and $\VD_P^*(Q')\cap \triangle$ into $\VD_P^*(Q\cup Q')\cap \triangle$ where $\triangle$ is the underlying triangle. 
In our algorithm, either $Q=S_\triangle$, $Q'=S(d_1)$, and $\eta=(d,d_2)$ or $Q=S_\triangle\cup S(d_1)$, $Q'=S(d_2)$, and $\eta=d$; in both cases, $S_\triangle\subseteq Q$.
The border will form on $\partial \triangle \setminus \eta$, i.e., $d_1$ for the former case and $(d_1, d_2)$ for the latter case. 
Although a wavefront only stores incomplete Voronoi cells, its associated diagram can still be accessed through the Voronoi edges of the stored cells. 
After the merge, new incomplete Voronoi edges will form, and their potential vertices will be created.

The Merge operation consists of two phases: (1) merge $\VD_P^*(Q)\cap \triangle$ and $\VD_P^*(Q')\cap \triangle$ into $\VD_P^*(Q\cup Q')\cap \triangle$
and (2) merge $W_\eta(Q)$ and $W_\eta(Q')$ into $W_\eta(Q\cup Q')$.

The first phase is to construct so-called \emph{merge curves}. A merge curve is a connected component consisting of border edges along $\partial \triangle \setminus \eta$ and Voronoi edges in $\VD_P^*(Q\cup Q')\cap \triangle$
associated with one site in $Q$ and one site in $Q'$; the ordering of mergve curves is the ordering of their appearances along $\eta$.
This phase is almost identical to the merge process by Papadopoulou and Lee \cite[Section~5]{PapadopoulouL98},
but since our data structure for a wavefront is different from theirs, 
a binary search along a wavefront to find a starting endpoint for a merge curve requires a different implementation. 
We first state this difference here, and will describe the details of tracing a merge curve in Section~\ref{subsub-op-merge}.

Assume $\eta$ to be oriented from $v_1$ to $v_{1,2}$ for $\eta=(d,d_2)$ and from $v_1$ to $v_2$ for $\eta=d$. 
By \cite[Lemma~4--6]{PapadopoulouL98}, 
a merge curve called \emph{initial} starts from $v_{1,2}$ for the latter case ($\eta=d$),
but all other merge curves have both endpoints on $\eta$, and those endpoints are associated with one site in $S_\triangle$. Let $Q_\triangle$ be the set of sites in $S_\triangle$ that have a wavelet in $W_{\eta}(Q)$.
If $\eta=(d,d_2)$, a site in $Q_\triangle$ can have two wavelets in $W_{\eta}(Q)$,
and with an abuse of terminology, such a site is imagined to have two copies, each to one wavelet. 
Since $Q_\triangle\subseteq Q$, finding a starting endpoint for each merge curve except the initial one is to test sites in $Q_\triangle$ following the ordering of their wavelets in $W_{\eta}(Q)$ along $\eta$. 
After finding a starting endpoint, the corresponding merge curve will be traced; when the tracing reaches $\eta$ again,
a stopping endpoint forms, and the first site in $Q_\triangle$ lying after the site inducing the stopping endpoint will be tested. 

Let $x$ be the next starting endpoint, which is unknown, and let $s$ be the next site in $Q_\triangle$ to test.
A two-level binary search on $W_\eta(Q')$ determines if $s$ induces $x$,
and if so, further determines the site $t\in Q'$ that induces $x$ with $s$ as well as the corresponding anchor.

The first-level binary search executes on the \BST{} for the wavelets in $W_\eta(Q')$,
and each step determines for a site $q\in Q'$ if its cell lies before or after $t$'s cell along $\eta$ or if $q=t$.
Let $y_1$ and $y_2$ denote the two intersection points between $\eta$ and the two incomplete edges of $s$ (in $W_\eta(Q)$), where $y_1$ lies before $y_2$ along $\eta$,
and let $z_1$ and $z_2$ be the two points defined similarly for $q$ (in $W_\eta(Q')$).
Since $s$ lies in $\triangle$, the distance between $s$ and any point in $\eta$ can be computed in $O(1)$ time.
The two incomplete edges of $q$ will be updated until $z_1$ and $z_2$,
so that the distance from $q$ to $z_1$ (resp.~to $z_2$) can be computed from the corresponding anchor.
For example, if $u$ is the anchor of the subcell that contains $z_1$, $d(z_1, q)=|\overline{z_1u}|+d(u,q)$.

The determination considers four cases. (Assume $s$ and $t$ induce the ``starting'' endpoint.)
\begin{itemize}
	\item $z_2$ lies before $y_1$ (resp.~$z_1$ lies after $y_2$): $t$'s cell lies after (resp.~before)  $q$'s cell.
	\item $z_1$ lies before $y_1$ and $z_2$ lies between $y_1$ and $y_2$ (resp.~$z_2$ lies after $y_2$ and $z_1$ lies between $y_1$ and $y_2$):
	if $z_2$ is closer to $q$ than to $s$ (resp.~$z_1$ is closer to $q$ than to $s$), then $t$'s cell lies after (resp.~before) $q$'s cell;
	otherwise, $t$ is $q$. 
	\item Both $y_1$ and $y_2$ lie between $z_1$ and $z_2$: $t$ is $q$.
	\item Both $z_1$ and $z_2$ lie between $y_1$ and $y_2$: let $x_s$ be the projection point of $s$ onto $\eta$.
	\begin{itemize}
		\item If $x_s$ lies before $z_1$, then $t$'s cell lies before $q$'s cell.
		\item If $x_s$ lies after $z_2$: if $z_2$ is closer to $q$ than to $s$, $t$'s cell lies after $q$'s cell;
		if both $z_1$ and $z_2$ are closer to $s$ than to $q$, $t$'s cell lies before $q$'s cell; otherwise, $t=q$.
		\item If $x_s$ lies between $z_1$ and $z_2$: if $z_1$ is closer to $s$ than to $q$, $t$'s cell lies before $q$'s cell; otherwise, $t=q$.
	\end{itemize}
\end{itemize}
If the first-level search does not find $t$, then $s$ does not induce the next starting endpoint $x$.

The second-level binary search executes on the \BST{} for $t$'s anchor list to either determine the next starting endpoint $x$ and $t$'s corresponding anchor 
or report that $s$ does not induce $x$. Let $u$ be the current anchor of $t$ to test,
and let $x_s$ be the projection point of $s$ onto $\eta$.
$u$'s ``interval'' on $\eta$ can be decided by checking $u$'s two neighboring anchors.
If $u$'s interval lies after $x_s$,
$x$ lies before $u$'s interval;
otherwise,
if both endpoints of $u$'s interval are closer to (resp.~farther from) $t$ than to (resp.~from) $s$, 
$x$ lies after (resp.~before) $u$'s interval,
and if one endpoint is closer to $t$ than to $s$ but the other is not, then $x$ lies in $u$'s interval and can be computed in $O(1)$ time since $d(x, s)=|\overline{xu}|+d(u,t)$.
If the second-level binary search does not find such an interval, $s$ does not induce the next starting endpoint $x$.

The second phase (i.e., merging $W_\eta(Q)$ and $W_\eta(Q')$ into $W_\eta(Q\cup Q')$)  splits $W_\eta(Q)$ and $W_\eta(Q')$ at the endpoints of merge curves,
and concatenates active parts at these endpoints where a part is called active if it contributes to $W_\eta(Q\cup Q')$.
In fact, the active parts along $\eta$ alternately come from $W_\eta(Q)$ and $W_\eta(Q')$.
At each merging endpoint, the two cells become adjacent, generating a new incomplete Voronoi edge. 
Potential vertices of these incomplete Voronoi edges will be computed and inserted into the corresponding priority queues. 
For each ending polygon vertex of $\eta$, 
if it is reflex but has not yet been an anchor of $W_\eta(Q\cup Q')$,
it will be inserted into its located wavelet as the first or the last anchor.


The total operation time is $O\big(|S_\triangle|(\log n+\log m)+\PV_{\mbox{new}}(\log m+\log^2 n)+I_{\mbox{new}}\big)$ plus amortized $O(1)$,
where $\PV_{\mbox{new}}$ is the number of created potential vertices and $I_{\mbox{new}}$ is the number of created diagram vertices while merging the two diagrams.
First, since each two-level binary search takes $O(\log m+\log n)$ time, finding starting points takes $O\big(|S_\triangle|(\log m+\log n)\big)$ time. 
Second, by Section~\ref{subsub-op-merge} or by \cite[Section~5]{PapadopoulouL98},
tracing a merge curve takes time linear in the number of deleted and created vertices, but the time to delete vertices has been charged at their creation,
implying that tracing all the merge curves takes $O(I_{\mbox{new}})$ time. 

Third, an incomplete Voronoi edge generates at least one potential vertex,
so the number of new incomplete Voronoi edges is $O(\PV_{\mbox{new}})$.
Since an endpoint of a merge curve corresponds to a new incomplete Voronoi edge,
there are $O(\PV_{\mbox{new}})$ split and $O(\PV_{\mbox{new}})$ concatenation operations, and since each operation takes $O(\log m+\log n)$ time,
it takes $O\big(\PV_{\mbox{new}}(\log m+\log n)\big)$ time to merge the two wavefronts.
By the same analysis in Section~\ref{sub-op-extend},
creating $\PV_{\mbox{new}}$ potential vertices takes $O\big(\PV_{\mbox{new}}(\log m+\log^2 n)\big)$ time.
Finally, inserting an ending polygon vertex of $\eta$ as an anchor takes amortized $O(1)$ time.

\subsubsection{Tracing merge Curves in Merge Operation}\label{subsub-op-merge}

We make some further definitions. The bisector $B(Q,Q')$ between $Q$ and $Q'$ is
the collection of points with the same minimum geodesic distance from both $Q$ and $Q'$, namely 
$B(Q, Q')=\{x\in P\mid \min_{q\in Q} d(x,q) = \min_{q'\in Q'} d(x,q')\}$.
In our algorithm, each anchor $u$ refers to either a reflex polygon vertex of $P$ or a point site in $S$;
we store its associated site $s(u)$, its geodesic distance $w(u)$ from $s(u)$,
and its anchor $a(u)$ with respect to $s(u)$. The \emph{weighted distance} $d_w(x,u)$ from a point $x$ to $u$ 
is $|\overline{xu}|+w(u)$, and the \emph{weighted bisector} between two anchors $u_1$ and $u_2$ is $\{x\in P\mid d_w(x,u_1)=d_w(x, u_2)\}$.

The process to trace a merge curve from the starting endpoint is identical to that presented by Papadopolou and Lee. Note that when a wavelet (incomplete cell) is visited, its incomplete edges will be updated to $\eta$, enabling the tracing between (incomplete) subcells.
The merge curve begins with the weighted bisector between the two anchors that induce the starting endpoint.
When the merge curve changes the underlying subcell in one of the two diagrams,
it continues with the weighted bisector between the new anchor and the other original anchor. 
When the merge curve hits $\eta'$, the tracing continues along $\eta'$ in the direction along which the next point is closer to the current site in $Q$ than to the current site in $Q'$;
until reaching an equidistant point, it turns into the interior of $\triangle$ again following the corresponding weighted bisector.
The merge curve may visit $\eta'$ several times. 
Finally, it reaches $\eta$  at the stopping endpoint. 
Border vertices form on $\eta'$ when the merge curve enters or leaves $\eta$ and when the merge curve changes the underlying subcell in one side of $\eta$.

Since tracing a merge curve takes time proportional to the number of deleted and created vertices but the time to delete vertices has been charged at their creation,
tracing all the merge curves takes $O(I_{\mbox{new}})$ time in total,
where $I_{\mbox{new}}$ is the number of newly created diagram vertices.

\subsection{Join operation}\label{sub-op-join}

A Join operation joins $W_d\big(S(d_3)\cup S_{\triangle'}\big)$ and $W_d\big(S'(d_4)\big)$ into $W_d\big(S'(d)\big)$ with building the border on $d$ but without merging $\VD_P^*\big(S(d_3)\cup S_{\triangle'}\big)\cap \triangle'$ and $\VD_P^*\big(S'(d_4)\big)\cap \triangle'$ into $\VD_P^*\big(S'(d)\big)\cap\triangle'$.
In the construction of $\SD'$ (Section~\ref{sub-preorder}),
our algorithm will actually join $W_{d_1}\big(S(d_2)\cup S_\triangle\big)$ and $W_{d_1}(S'(d))$ into $W_{d_1}(S'(d_1))$, and join $W_{d_2}\big(S(d_1)\cup S_\triangle\big)$ and $W_{d_2}(S'(d))$ into $W_{d_2}(S'(d_2))$.
To interpret a Join operation, for the former case, one would replace $d$, $d_3$, $d_4$, and $\triangle'$ with $d_1$, $d_2$, $d$, and $\triangle$, respectively, 
and for the latter case, one would replace $d$, $d_3$, $d_4$, and $\triangle'$ with $d_2$, $d_1$, $d$, and $\triangle$, respectively.

The reason for not merging $\VD_P^*\big(S(d_3)\cup S_{\triangle'}\big)\cap \triangle'$ and $\VD_P^*\big(S'(d_4)\big)\cap \triangle'$ into $\VD_P^*\big(S'(d)\big)\cap\triangle'$  is that $S(d_3)\cup S_{\triangle'}$ contributes nothing to $\SD'\cap \triangle'=\VD_P^*\big(S'(d_4)\big)\cap \triangle'$.
The border on $d$ while merging $W_d\big(S(d_3)\cup S_{\triangle'}\big)$ and $W_d\big(S'(d_4)\big)$ into $W_d\big(S'(d)\big)$ is the collection of points in $d$ that are closer to $S(d_3)\cup S_{\triangle'}$ than to $S'(d_4)$, namely $\{x\in d\mid \min_{s\in S(d_3)\cup S_{\triangle'}} d(x, s) \leq  \min_{s'\in S'(d_4)}d(x, s')\}$.

At a high level, a Join operation first builds the border on $d$, and then joins the two wavefronts according to the border.
This operation is conceptually identical to the process in \cite[Section~6]{PapadopoulouL98}
except that our data structure of a wavefront is different. 

With a slight abuse of terminology, let $b'(d)$ denote the collection of connected components of border edges on $d$ in $SD'$.\footnote{In \cite[Section~6]{PapadopoulouL98}, the authors define $\sigma'(d)$ as the Voronoi edges in $SD'\cap \triangle$ that are associated with one site in $S'(d_4)$ and one site in $S(d_3)\cup S_{\triangle'}$, but only compute $\sigma'(d)\cap d$, which is exactly $b'(d)$. }
These components in $b'(d)$ are called \emph{joint segments} and are ordered from $v_2$ to $v_1$. 
By \cite[Lemma~14--15]{PapadopoulouL98},
$b'(d)$ satisfies the following conditions:
\begin{itemize}
	\item At most one joint segment starts at $v_2$, and such a joint segment is called \emph{initial}.
	\item Except the initial one, both endpoints of each joint segment are associated with a site in $S_{\triangle'}$. The first endpoint is called \emph{starting},
	and the second endpoint is called \emph{stopping}.
\end{itemize}
The second condition results from the fact that the Voronoi cells associated with $S(d_3)$ and $S(d_4)$ do not ``cross'' each other along $d$,
and this desirable condition enables locating a starting endpoint of a non-initial joint segment by searching $W_d\big(S'(d_4)\big)$ with sites in $S_{\triangle'}$.

Assume that $d$ is oriented from $v_2$ to $v_1$.
Joint segments will be constructed one by one from $v_2$ to $v_1$.
For each joint segment (except the initial one), its starting endpoint is first located on $d$ by searching $W_d\big(S'(d_4)\big)$ and then it is traced from its starting endpoint along $d$ until reaching its stopping endpoint. We remark that the initial joint segment will be directly traced from $v_2$.

Let $S_{\triangle'|d}$ be the set of sites in $S_{\triangle'}$ that have a wavelet in $W_d\big(S(d_3)\cup S_{\triangle'}\big)$. 
By the second condition, 
finding the starting  endpoint of a joint segment is to ``locate'' sites of $S_{\triangle'|d}$ in $W_d\big(S'(d_4)\big)$ since each endpoint of a joint segment (except the initial one)
is associated with a site in $S_{\triangle'|d}$.
Therefore,
sites in $S_{\triangle'|d}$ will be tested following the ordering of their wavelets in $W_d\big(S(d_3)\cup S_{\triangle'}\big)$ from $v_2$ to $v_1$;
after tracing a joint segment, the test will continue on the next site in $S_{\triangle'|d}$, which lies after the site in $S_{\triangle'|d}$ associated the last traced stopping endpoint. 
Note that each site in $S_{\triangle'|d}$ has exactly one wavelet in $W_d\big(S(d_3)\cup S_{\triangle'}\big)$.

Let $x$ be the next starting endpoint, which is unknown, and let $s$ be the site in $S_{\triangle'|d}$ to test.
A two-level binary search on $W_d\big(S'(d_4)\big)$, which is identical to the one in Section~\ref{sub-op-merge}, can determine if $s$ induces $x$,
and if so, determine the site $t\in S'(d_4)$ that induces $x$ with $s$ together with the corresponding anchor.
The total time to find starting points is trivially $O\big(|S_{\triangle'}|(\log n+\log m)\big)$.

The process to trace a joint segment from the starting endpoint is simpler than tracing a merge curve in Section~\ref{sub-op-merge}.
The process walks along $d$ from the starting endpoint. 
Every time when one of the two corresponding anchors in $W_d\big(S(d_3)\cup S_{\triangle'}\big)$
and $W_d\big(S'(d_4)\big)$ changes, a border vertex is created accordingly. The process ends when reaching a point equidistant from the site of its located wavelet in $W_ d\big(S(d_3)\cup S_{\triangle'}\big)$ and the site of its located in $W_d\big(S'(d_4)\big)$; 
this point is a border vertex and also the second endpoint of the joint segment. 
Note that the change of an anchor can be determined by the anchor lists, and the distance from a site to a point can be determined through the corresponding anchor.
The time to trace all the joint segments is proportional to the number of created border vertices,
i.e., $O(I_{\mbox{new}})$ time in total, where $I_{\mbox{new}}$ is the number of newly created diagram vertices.

After constructing joint segments, we split each of $W_d\big(S(d_3)\cup S_{\triangle'}\big)$ and $W_d\big(S'(d_4)\big)$ at the endpoints of these joint segments, and concatenate the active parts at these endpoints into $W_d\big(S'(d)\big)$, where ``active'' means having a wavelet in $W_d\big(S'(d)\big)$.
For each ending polygon vertex of $d$, i.e., $v_1$ or $v_2$,  
if it is reflex but has not been an anchor of $W_d\big(S'(d)\big)$,
it will be inserted into its located wavelet as the first or the last anchor.

The time to update the wavefronts is $O\big(\PV_{\mbox{new}}(\log m+\log^2 n)\big)$ plus amortized $O(1)$, where $K_{\mbox{new}}$ is the number of created potential vertices.
First, an incomplete Voronoi edge generates at least one potential vertex,
so the number of new incomplete Voronoi edges is $O(\PV_{\mbox{new}})$.
Since an endpoint of a joint segment corresponds to a new incomplete Voronoi edge,
there are $O(\PV_{\mbox{new}})$ split and $O(\PV_{\mbox{new}})$ concatenation operations, and since each operation takes $O(\log m+\log n)$ time,
it takes $O\big(\PV_{\mbox{new}}(\log m+\log n)\big)$ time to join the two wavefronts.
Moreover, by the same analysis in Section~\ref{sub-op-extend},
creating $\PV_{\mbox{new}}$ potential vertices takes $O\big(\PV_{\mbox{new}}(\log m+\log^2 n)\big)$ time.
Finally, inserting $v_1$ or $v_2$ as an anchor takes amortized $O(1)$ time.

To conclude, a Joint operation takes  $O\big((|S_{\triangle'}|(\log n+\log m)+\PV_{\mbox{new}}(\log m+\log^2 n)+I_{\mbox{new}}\big)$ plus amortized $O(1)$ time.

\subsection{Split operation}\label{sub-op-split}

A Split operation splits a wavefront associated with a pair of diagonals at the common polygon vertex using a binary search. 
First, the operation locates the wavelet that contains the common polygon vertex and the corresponding anchor. 
Second, the operation splits the wavelet, i.e., the corresponding list of anchors, at the located anchor, and duplicates the located anchor since it appears in both resulting wavelets.
Third, the operation splits the wavefront between the two resulting wavelets. 
Finally, if the common polygon vertex is reflex, the operation inserts it into both of the duplicate wavelets as an anchor. 

The operation time is $O(\log n+\log m)$.
Since a wavefront has $O(m)$ wavelets and a wavelet has $O(n)$ anchors, 
both locating the common polygon vertex and spliting the wavelet and wavefront take $O(\log n+\log m)$ time.
Although inserting the common polygon vertex takes amortized $O(1)$ time, since the worst-case time to insert an anchor is $O(\log n)$,
the time to insert the common polygon vertex is dominated by the time-factor $O(\log n+\log m)$.
Since there is no new site, there is no new incomplete Voronoi edge and no new potential vertex.

\subsection{Divide operation}\label{sub-op-divide}

A Divide operation divides a wavefront associated with a pair of diagonals at the common polygon vertex by traversing one diagonal instead of using a binary search. 
Although a Divide operation seems a brute-force way compared to a Split operation,
since a Split operation takes $\Omega(\log n + \log m)$ time 
and there are $\Omega(n)$ events to separate a wavefront, 
if only Split operations are adopted,
the total construction time would be $\Omega\big(n(\log n+\log m)\big)$.

First, the wavefront is traversed from the end of the selected diagonal subcell by subcell, i.e., anchor by anchor, until reaching the common vertex.
Then, the wavefront is separated at the common polygon vertex by removing all the visited anchors except the last one,
duplicating the last one, 
and building a new wavefront for these ``removed'' anchors and the duplicate anchor from scratch.
Finally, if the common polygon vertex is reflex, it is inserted into its located wavelets in both resultant wavefronts as an anchor without a binary search (since it is the first or last anchor of its located wavelets).

The total operation time is amortized $O(A_{\mbox{vis}}+1)$, 
where $A_{\mbox{vis}}$ is the number of visited anchors.
Since each wavelet (resp.~anchor) records its two neighboring wavelets (resp.~anchors) in the augmented \BST{}, the time to locate the common polygon vertex is $O(A_{\mbox{vis}})$. Recall that a cell must have one subcell, and the time to remove a wavelet or an anchor has been charged when it was inserted. Finally,
building the new wavefront from scratch takes amortized $O(A_{\mbox{vis}})$ time,
and inserting the common vertex takes amortized $O(1)$ time.

\subsection{Insert operation}\label{sub-op-insert}

For a triangle $\triangle$, an Insert operation inserts $S_\triangle$ into $W_{d_1}\big(S(d_2)\big)$ to form $W_{d_1}\big(S(d_2)\cup S_\triangle\big)$.
An Insert operation executes in the construction of $\SD$, but its outcome will be used to construct $\SD'$. 
Let $T$ denote the intermediate site set during the Insert operation, so that $T=S(d_2)$ at the beginning. 

For each site $s\in S_\triangle$, let $x_s$ be its vertical projection point onto $d_1$, and conduct a binary search on the wavefront $W_{d_1}(T)$ to locate the anchor whose subcell contains $x_s$. During the binary search, the two incomplete Voronoi/polygonal edges of each visited wavelet will be updated up to $d_1$. 
Recall that the \emph{weighted distance} between a point $y$ and an anchor $u$ associated with a site $t$ 
is $|\overline{yu}|+d(u,t)$.

If $x_s$ is closer to $s$ than to the located anchor under the weighted distance,
traverse anchors in $W_{d_1}(T)$ from $x_s$ along each direction of $d_1$ until either touching a polygon vertex of $d_1$ or reaching a point equidistant from $s$ and the current visited anchor under the weighted distance.
If all points in the ``interval'' of a visited anchor on $d_1$ are fully closer to $s$ than to the anchor, remove the anchor from its associated wavelet,
and if all the anchors of a wavelet have been removed, remove the wavelet from $W_{d_1}(T)$.
After the traversal, insert $s$ into $W_{d_1}(T)$, namely create and insert its wavelet into  $W_{d_1}(T)$.
Since $S(d_2)$ and $S_\triangle$ lie on the same side of $d_1$ and all sites in $S(d_2)$ lie outside $\triangle$,
each cell of a site in $S(d_2)$ will not be separated by a cell of a site in $S_\triangle$ ``along $d_1$'', supporting the correctness of the above process.

After processing all the sites in $S_\triangle$, check the two polygon vertices, $v_1$ and $v_{1,2}$, of $d_1$, and if $v_1$ (resp.~$v_{1,2}$) should belong to a wavelet associated a site in $S_\triangle$,
insert $v_1$ (resp.~$v_{1,2}$) into the wavelet as the first or last anchor. For each inserted wavelet, generate its incomplete Voronoi or polygonal edges, and compute the corresponding potential vertices.

The total operation time is $O(|S_\triangle|(\log m+\log^2n))$.
First, since each binary search takes $O(\log n+\log m)$ time, the total time for binary searches is $O\big(|S_\triangle|(\log n + \log m)\big)$.
Second, inserting a wavelet takes $O(\log m)$ time, and the time to remove anchors and wavelets has been already charged at their insertion,
leading to $O(|S_\triangle|\log m)$ time.
Third, although inserting a polygon vertex at a known position takes amortized $O(1)$ time, 
since $d_1$ has only two polygon vertices, the amortized time to insert them is dominated by the time-factor $|S_\triangle|\log m$. Note that an Insertion operation occurs only if $S_\triangle\neq \emptyset$. 
Finally, since there are at most $|S_\triangle|$ new wavelets in $W_{d_1}\big(S(d_2)\cup S_\triangle\big)$,
there are $O(|S_\triangle|)$ new incomplete Voronoi edges,
implying that the time to create new potential vertices is $O(|S_\triangle|(\log m+\log^2n))$.

\subsection{Propagate operation}\label{sub-op-propagate}


A Propagate operation propagates a wavefront $W_{d}\big(S'(d)\big)$ into $P(d)$, i.e., from the upside to the downside of $d$,
to build $\SD'\cap P(d)$ provided that $S\cap P(d)=S(d)=\emptyset$.
Since $S(d)=\emptyset$, then $\SD'\cap P(d)$ is exactly $\VD^*_P(S)\cap P(d)$.
To some extent, a Propagate operation is a generalized version of an Extend operation underlying a sub-polygon instead of a triangle.

The operation consists of two phases. The first phase constructs $\VD_P(S)\cap P(d)$,
and the second phase refines each cell into subcells to obtain $\VD^*_P(S)\cap P(d)$.

The first phase ``sweeps'' the triangles in $P(d)$ by a preorder traversal of the subtree of $\T$ rooted at $\triangle(d)$. 
Similar to the Extend operation,
this sweep processes potential vertices inside each triangle to construct Voronoi edges and to update the wavefront accordingly.
However, this sweep will not process the polygon vertices of $P(d)$, 
so that the anchor lists will not be updated, preventing constructing a Voronoi edge using the two corresponding anchor lists. 
Fortunately,
Oh and Ahn~\cite{OhA17} gave another technique that obtains in $O(\log n)$ time the two anchor lists of a Voronoi edge provided that the two Voronoi vertices are given, 
so a Voronoi edge can still be built in time proportional to $\log n$ plus the number of its breakpoints.

Therefore, the first phase takes $O\big(\PV_{d}(\log m+\log^2 n)+|\VD_P(S)\cap P(d)|\big)$ time,
where $\PV_{d}$ is the number of involved potential vertices.
Note that for the Voronoi edges that intersect $d$, their breakpoints outside $P(d)$ will be counted in the respective triangles in $P'(d)$.

The second phase triangulates each cell in $\VD_P(S)\cap P(d)$ and constructs the SPM from the associated site in each triangulated cell. 
Since Chazelle's algorithm~\cite{Chazelle91} takes linear time to triangulate a simple polygon
and Guibas et al's algorithm~\cite{GuibasHLST87} takes linear time to build the SPM in a triangulated polygon,
the second phase takes $O(|\SD'\cap P(d)|)$ time.

\begin{remark}
Although all the sites lie outside $P(d)$, 
the information stored in $W_{d}\big(S'(d)\big)$
allows us not to conduct Guibas et al's algorithm from scratch.
For example, for each anchor $u$, its anchor  $a(u)$ is also stored, 
and the SPM edge of $u$ is the line segment from $u$ to the polygon boundary along the direction from $a(u)$ to $u$.
\end{remark}

To sum up, a Propagate operation takes $O\big(\PV_{d}(\log m+\log^2 n)+|\SD'\cap P(d)|\big)$ time.

\section{Subdivision construction}\label{sec-construction}

\subsection{Construction of $\SD$}\label{sub-postorder}

To construct $\SD$, we process each triangle $\triangle$ by the postorder traversal of the rooted partition tree $\T$ and build $\SD\cap \triangle$.
We first assume that no diagonal of $\triangle$ is a polygon side, and we will discuss the other cases later. 
Let $d$ be the root diagonal of $\triangle$ and adopt the convention in Section~\ref{sub-pre-triangulation} and Section~\ref{sub-pre-SD}.
When processing $\triangle$, since its two children, $\triangle(d_1)$ and $\triangle(d_2)$, have been processed,
$W_{d_1}\big(S(d_1)\big)$ and $W_{d_2}\big(S(d_2)\big)$ are available.

The processing of each triangle $\triangle$ consists of 8 steps:
\begin{enumerate}
	\item Initiate $\VD_\triangle(S_\triangle)$ and $W_{(d,d_2)}(S_\triangle)$.
	\item Extend $W_{d_1}\big(S(d_1)\big)$ into $\triangle$ to generate $W_{(d,d_2)}\big(S(d_1)\big)$ and construct $\VD_P^*\big(S(d_1)\big)\cap \triangle$. 
	\item Merge $W_{(d,d_2)}(S_\triangle)$ and $W_{(d,d_2)}\big(S(d_1)\big)$ into $W_{(d,d_2)}\big(S(d_1)\cup S_\triangle\big)$ by which $\VD_\triangle(S_\triangle)$ and $\VD_P^*\big(S(d_1)\big)\cap \triangle$ are merged into $\VD_P^*\big(S(d_1)\cup S_\triangle\big)\cap \triangle$.
	\item Divide $W_{(d,d_2)}\big(S(d_1)\cup S_\triangle \big)$ into $W_{d}\big(S(d_1)\cup S_\triangle\big)$ and $W_{d_2}\big(S(d_1)\cup S_\triangle\big)$ along  $d_2$. 
	\item Extend $W_{d_2}\big(S(d_2)\big)$ into $\triangle$ to generate $W_{(d_1, d)}\big(S(d_2)\big)$ and construct $\VD_P^*\big(S(d_2)\big)\cap \triangle$. 
	\item Divide $W_{(d_1, d)}\big(S(d_2)\big)$ into $W_{d_1}\big(S(d_2)\big)$ and $W_{d}\big(S(d_2)\big)$ along $d_1$. 
	\item Insert $S_\triangle$  into $W_{d_1}\big(S(d_2)\big)$ to obtain $W_{d_1}\big(S(d_2)\cup S_\triangle \big)$. 
	\item Merge $W_{d}\big(S(d_1)\cup S_\triangle\big)$ and $W_{d}\big(S(d_2)\big)$ into $W_{d}\big(S(d)\big)$ by which $\VD_P^*\big(S(d_1)\cup S_\triangle\big)\cap \triangle$ and  $\VD_P^*\big(S(d_2)\big)\cap \triangle$ are merged into $\VD_P^*\big(S(d)\big)\cap \triangle=\SD\cap \triangle$. 
\end{enumerate}
We remark that $W_{d_1}\big(S(d_2)\cup S_\triangle \big)$ and $W_{d_2}\big(S(d_1)\cup S_\triangle\big)$ will be used to construct $\SD'$.

If exactly one diagonal $d$ of $\triangle$ is not a polygon side, it is either the root triangle $\blacktriangle$ or a leaf triangle. 
For the former, compute $\VD_\blacktriangle(S_\blacktriangle)$, extend $W_{d}(S(d))$ into $\blacktriangle$ to build $\VD_P^*(S(d))\cap \blacktriangle$,
and merge $\VD_\blacktriangle(S_\blacktriangle)$ and $\VD_P^*(S(d))\cap \blacktriangle$ into $\VD_P^*(S)\cap \blacktriangle=\SD\cap \blacktriangle$;
for the latter, compute $\VD_{\triangle}(S_\triangle)$ and initiate $W_{d}(S_\triangle)$.
If exactly two diagonals, $d$ and $d'$, of $\triangle$ are not polygon sides, where $d$ is the root diagonal,
then compute $\VD_\triangle(S_\triangle)$ to initiate $W_{d}(S_\triangle)$ and $W_{d'}(S_\triangle)$,
extend $W_{d'}\big(S(d')\big)$ into $\triangle$ to obtain $W_{d}\big(S(d')\big)$ and $\VD_P^*\big(S(d')\big)\cap \triangle$, and merge $W_{d}(S_\triangle)$ and $W_{d}\big(S(d')\big)$ to obtain $W_{d}\big(S(d)\big)$ and $\VD_P^*\big(S(d)\big)\cap\triangle=\SD\cap \triangle$. 

By the operation times in Section~\ref{sec-operation},
the time to process $\triangle$ is summarized as follows. 

\begin{lemma}\label{lem-time-triangle-post}
	It takes $O\big((|S_\triangle|+\PV_\triangle)(\log m+\log^2n)+I_\triangle\big)$ plus amortized $O(A_\triangle+1)$ time to process $\triangle$,
	where $\PV_\triangle$ is the number of involved potential vertices,
	$I_\triangle$ is the number of created diagram vertices,
	and $A_\triangle$ is the number of visited anchors in Steps~4 and 6.
\end{lemma}
\begin{proof}
	Step~1 takes $O\big(|S_\triangle|(\log m+\log^2n)\big)$ time;
	Step~2 and Step~5 take $O\big(\PV_\triangle(\log m+\log^2n)\big)+I_\triangle)$ plus amortized $O(1)$ time; 
	Step~3 and Step~8 take $O\big(|S_\triangle|(\log n+\log m)+\PV_\triangle(\log m+\log^2n)+I_\triangle\big)$ plus amortized $O(1)$ time; Step~7 takes $O\big(|S_\triangle|(\log m+\log^2 n)\big)$ time; Step~4 and Step~6 take amortized $O(A_\triangle+1)$ time.
\end{proof}

To apply Lemma~\ref{lem-time-triangle-post}, we bound $\sum_\triangle \PV_\triangle$,  $\sum_\triangle I_\triangle$  and $\sum_\triangle A_\triangle$ by the following two lemmas.

\begin{lemma}\label{lem-number-potential-post}
	In the construction of $\SD$, $\sum_\triangle \PV_\triangle=O(m)$ and $\sum_\triangle I_\triangle=O(n+m)$. 
\end{lemma}
\begin{proof}
	There are 6 intermediate diagrams:
	$\VD_\triangle(S_\triangle)$ (step~1), $\VD_P^*\big(S(d_1)\big)\cap\triangle$ (step~2),
	$\VD_P^*\big(S(d_1)\cup S_\triangle\big)\cap \triangle$ (step~3), $\VD_P^*\big(S(d_2)\big)\cap\triangle$ (step~5), 
	$\VD_P^*\big(S(d_2)\cup S_\triangle\big)\cap\triangle$ (step~7),
	and  $\VD_P^*\big(S(d)\big)\cap \triangle=\SD\cap \triangle$ (step~8).
	First, a potential vertex arises due to the formation of an incomplete Voronoi edge, and an incomplete Voronoi edge 
	generates $O(1)$ potential vertices. 
	For the first diagram, the total number of Voronoi edges is $O(\sum_\triangle |S_\triangle|)=O(|S|)=O(m)$.
	For each of the other 5 diagrams, we can define borders in a similar way to $\SD$ and thus obtain a subdivision of $P$. 
	Since each site has at most one cell in each resultant subdivision,
	Euler's formula implies that each subdivision contains $O(m)$ Voronoi edges among cells, leading to the conclusion that $\sum_\triangle \PV_\triangle=O(m)$. 
	Second, a created diagram vertex must be a vertex of the first diagram or the other 5 subdivisions. 
	Since there are $m$ sites, the first diagram results in $O(m)$ vertices for all the triangles,
	and by the same reasoning of \cite[Lemma 3]{PapadopoulouL98},
	each subdivision has $O(n+m)$ vertice, 
	leading to that $\sum_\triangle I_\triangle=O(n+m)$.
\end{proof}

\begin{lemma}\label{lem-anchor-post}
	In the construction of $\SD$, $\sum_\triangle A_\triangle=O(n+m)$. 
\end{lemma}
\begin{proof}
	We consider Step~4 (divide $W_{(d,d_2)}\big(S(d_1)\cup S_\triangle \big)$ along $d_2$),
	which is similar to Step~6. Since there are $O(n+m)$ anchors,
	it is sufficient to bound the number of anchors that are visited by Step~4 but still involved in the future construction of $\SD$, namely $\SD\cap P'(d)$.
	Since the subcell of each visited anchor intersects $d_2$,
	if the subcell does not intersect $d$,
	its anchor will not be involved in constructing $\SD\cap P'(d)$.
	A subcell of a visited anchor intersects $d$ in two cases.
	In the first case,
	the subcell contains $v_2$  and thus intersects both $d$ and $d_2$. 
	Since there are $O(n)$ triangles, 
	the total number for the first case is $O(n)$.
	In the second case, the subcell intersects both $d$ and $d_2$ but does not contain $v_2$.
	By the definition of $W_{(d,d_2)}\big(S(d_1)\cup S_\triangle \big)$, its associated ``site'' belongs to either $S_\triangle$ or $S(d_1)$. 
	For the former, since its anchor must be the site itself, 
	the total number is $\sum_\triangle S_\triangle=m$. 
	For the latter, since all the sites in $S(d_1)$ lie outside $\triangle$,
	only one site in $S(d_1)$ can own a cell intersecting both $d$ and $d_2$. 
	Moreover, due to the visibility of a subcell,
	only one subcell in such a cell can intersect both $d$ and $d_2$.
	Since there are $O(n)$ triangles, the total number is $O(n)$.
\end{proof}

By Lemma~\ref{lem-time-triangle-post}, ~\ref{lem-number-potential-post}, and ~\ref{lem-anchor-post}, we conclude the construction time of $\SD$ as follows.

\begin{theorem}\label{thm-time-post}
	$\SD$ can be constructed in $O\big(n+m(\log m+\log^2 n)\big)$ time.
\end{theorem}
\begin{proof}
	By Lemma~\ref{lem-time-triangle-post}, we need to bound $\sum_\triangle \big((|S_\triangle|+\PV_\triangle)\cdot (\log m+\log^2n)+I_\triangle+A_\triangle+1\big)$. 
	It is trivial that $\sum_\triangle |S_\triangle|=|S|=m$.
	By Lemma~\ref{lem-number-potential-post} and Lemma~\ref{lem-anchor-post}, 
	$\sum_\triangle \PV_\triangle=O(m)$, $\sum_\triangle I_\triangle=O(n+m)$, and $\sum_\triangle A_\triangle=O(n+m)$,
	leading to the statement.
\end{proof}

\subsection{Construction of $\SD'$}\label{sub-preorder}

To construct $\SD'$, we processes each triangle by the preorder traversal of the partition tree $\T$.
For the root triangle $\blacktriangle$, 
let $\tilde{d}$ be its diagonal that is not a polygon side, 
build $W_{\tilde{d}}(S_\blacktriangle)=W_{\tilde{d}}\big(S'\big(\tilde{d}))$,
and if $S_\blacktriangle=S$, further propagate $W_{\tilde{d}}(S_\blacktriangle)$ into $P(\tilde{d})$.
For other triangles $\triangle$, 
we assume that neither $\triangle$ nor its parent $\triangle'$ has a polygon side; 
the other cases can be processed in a similar way.
Since $\triangle'$ has been processed,
$W_d\big(S'(d)\big)$ is available,
and by the construction of $\SD$, $W_{d_2}\big(S(d_1)\cup S_\triangle\big)$ and  $W_{d_1}\big(S(d_2)\cup S_\triangle\big)$ have been generated.

If $S(d)\neq \emptyset$, $\triangle$ is processed by the following 4 steps:
\begin{enumerate}
	\item Extend $W_d\big(S'(d)\big)$ into $\triangle$ to obtain
	$W_{(d_1, d_2)}(S'(d))$ and $\VD_P^*\big(S'(d)\big)\cap \triangle=\SD'\cap\triangle$.
	\item  If neither $S(d_1)$ nor $S(d_2)$ is empty, 
	split $W_{(d_1, d_2)}(S'(d))$ into $W_{d_1}(S'(d))$ and $W_{d_2}(S'(d))$;
	       otherwise, if $S(d_1)$ (resp.~$S(d_2)$) is empty, 
	       divide $W_{(d_1, d_2)}(S'(d))$ into $W_{d_1}(S'(d))$ and $W_{d_2}(S'(d))$ along $d_1$ (resp.~$d_2$).
	\item Join $W_{d_1}\big(S(d_2)\cup S_\triangle\big)$ and $W_{d_1}(S'(d))$ into $W_{d_1}(S'(d_1))$, and join  $W_{d_2}\big(S(d_1)\cup S_\triangle\big)$ and $W_{d_2}(S'(d))$ into $W_{d_2}(S'(d_2))$.
	\item If $S(d_1)$ (resp.~$S(d_2)$) is empty, 
	propagate $W_{d_1}(S'(d_1))$ (resp.~$W_{d_2}(S'(d_2))$) into $P(d_1)$ (resp.~$P(d_2)$) to build $\SD'\cap P(d_1)$ (resp.~$\SD'\cap P(d_2)$).
	
\end{enumerate}

We first analyze Split and Propagate operations.

\begin{lemma}\label{lem-split-pre}
	The total number of Split operations is $O(m)$.
\end{lemma}
\begin{proof}
	A Split operation occurs only if neither $S(d_1)$ nor $S(d_2)$ is empty. 
	We recursively remove leaf nodes (triangles) containing no site from  $\T$, so that each leaf node in the resulting tree $\T'$ contains a site.
	Then a Split operation occurs in an internal node of $\T'$ with two children.
	Since there are $m$ sites, $\T'$ has $O(m)$ leaf nodes,
	and since $\T'$ is a binary tree, the number of internal nodes with two children is $O(m)$, 
	leading to $O(m)$ Split operations.
\end{proof}

\begin{lemma}\label{lem-propagate-pre}
The total time for all the Propagate operations is $O\big(n+m(\log m+\log^2 n)\big)$.
\end{lemma}
\begin{proof}
	For a sub-polygon to propagate, since each edge of the resultant subdivision has a vertex in the sub-polygon,
	the number of created edges is bounded by the number of created vertices.
	Therefore, 	by Section~\ref{sub-op-propagate}, the total time is $O\big(\PV(\log m+\log^2n)+I\big)$,
	where $\PV$ is the number of involved potential vertices and $I$ is the number of created diagram vertices.
    Moreover, by the same reasoning of Lemma~\ref{lem-number-potential-post}, $\PV=O(m)$,
	and since all the sub-polygons to propagate are disjoint, $I=O(|\SD'|)=O(m+n)$, leading to the statement.
\end{proof}

By Lemma~\ref{lem-split-pre}, Lemma~\ref{lem-propagate-pre} and the reasoning of Theorem~\ref{thm-time-post},
we conclude the construction time of $\SD'$ as follows.

\begin{theorem}\label{thm-time-pre}
	$\SD'$ can be constructed in $O\big(n+m(\log m+\log^2 n)\big)$ time.
\end{theorem}
\begin{proof}
	The analysis for Extend operations (Step~1) is identical to that in the construction of $\SD$,
	and the analysis for Join operations (Step~3) is similar to the analysis for Merge operations in the construction of $\SD$,
	both of which yield $O\big(n+m(\log m+\log^2 n)\big)$ time in total.
	For Split operations (Step~2),
	by Lemma~\ref{lem-split-pre}, there are $O(m)$ Split operations in total, 
	and by Appendix~\ref{sub-op-split}, each Split operation take $O(\log m+\log n)$ time, leading to $O\big(m(\log m+\log n)\big)$ time in total.
	For a Divide operation (Step~2), although each of the two resultant wavefronts could be joined with another wavefront in Step~3,
	since the wavefront along the traversed diagonal will propagate into the corresponding sub-polygon in Step~4 and will not separate anymore,
	the reasoning of Lemma~\ref{lem-anchor-post} implies that the total time for all the Divide operations is $O(n+m)$.  
	Finally,
	by Lemma~\ref{lem-propagate-pre}, the total time for all the Propagate operations (Step~4) is $O\big(n+m(\log m+\log^2 n)\big)$,
	leading to the statement.
\end{proof}

\newpage
\bibliographystyle{plain}
\bibliography{p-58-liu}

\newpage

 \end{document}